\documentclass[conference]{IEEEtran}
\IEEEoverridecommandlockouts

\usepackage{graphicx}
\usepackage{algorithm}
\usepackage{stackengine}
\usepackage{cite}
\usepackage{color}
\usepackage{xcolor}
\usepackage{amsmath,amsthm,amssymb,amsfonts,bm}
\usepackage{textcomp}
\usepackage{multirow}
\usepackage{url}

\def\BibTeX{{\rm B\kern-.05em{\sc i\kern-.025em b}\kern-.08em
    T\kern-.1667em\lower.7ex\hbox{E}\kern-.125emX}}

\usepackage[noend]{algpseudocode}

\makeatletter
\newcommand*{\rom}[1]{\expandafter\@slowromancap\romannumeral #1@}

\newtheorem{theorem}{Theorem}

\def \v x{\bm x}

\def \v x{\bm X}

\renewcommand{\v}[1]{\ensuremath{\boldsymbol{#1}}}

\begin{document}

\title{Optimal Privacy-Preserving Distributed Median Consensus
}

\author{\IEEEauthorblockN{Wenrui Yu}
\IEEEauthorblockA{
\textit{Aalborg University}\\
Denmark \\
wenyu@es.aau.dk}
\and
\IEEEauthorblockN{Qiongxiu Li}
\IEEEauthorblockA{\textit{Aalborg University}\\
Denmark \\
qili@es.aau.dk}
\and
\IEEEauthorblockN{Richard Heusdens}
\IEEEauthorblockA{\textit{Netherlands Defence Academy}\\
\textit{Delft University of Technology}\\
the Netherlands \\
r.heusdens@tudelft.nl}
\and
\IEEEauthorblockN{Sokol Kosta}
\IEEEauthorblockA{\textit{Aalborg University}\\
Denmark \\
sok@es.aau.dk}
}

\maketitle

\begin{abstract}
Distributed median consensus has emerged as a critical paradigm in multi-agent systems due to the inherent robustness of the median against outliers and anomalies in measurement. Despite the sensitivity of the data involved, the development of privacy-preserving mechanisms for median consensus remains underexplored. 
In this work, we present the first rigorous analysis of privacy in distributed median consensus, focusing on an $L_1$-norm minimization framework. We establish necessary and sufficient conditions under which exact consensus and perfect privacy—defined as zero information leakage—can be achieved simultaneously. Our information-theoretic analysis provides provable guarantees against passive and eavesdropping adversaries, ensuring that private data remain concealed. Extensive numerical experiments validate our theoretical results, demonstrating the practical feasibility of achieving both accuracy and privacy in distributed median consensus.
\end{abstract}


\section{Introduction}
\label{sec:intro}
Consensus algorithms facilitate agreement in distributed systems through localized computations and information exchange among neighboring nodes. Common examples include averaging~\cite{francca2020distributed,zhang2014asynchronous}, maximum/minimum~\cite{deplano2023unified}, and median consensus~\cite{deplano2023unified,franceschelli2016finite,dashti2019dynamic}. These algorithms form the backbone of distributed optimization and have been extensively applied in sensor networks, multi-agent systems, and high-performance computing~\cite{singhal2017lsc}.  Among them, median consensus stands out due to its inherent robustness to outliers---a critical advantage in large-scale sensor and multi-agent networks where noisy and unreliable data are common. By mitigating the influence of extreme values, median consensus is particularly well-suited for environments such as distributed computing clusters and parallel processing frameworks. However, despite its robustness, prior research has largely overlooked the issue of \emph{privacy preservation} in median consensus algorithms.

Several privacy-preserving frameworks have been proposed for consensus problems more broadly.  A widely applicable framework for ensuring privacy in consensus problems involves the use of differential privacy (DP) techniques \cite{huang2015differentially,nozari2018differentially,zhang2016dynamic,zhang2018improving,xiong2020privacy}. However, these methods often involve a trade-off between privacy preservation and the accuracy of the consensus output. Another approach to protecting privacy is the secure multiparty computation (SMPC) \cite{gupta2017privacy,li2019privacyA,tjell2020privacy,tjell2019privacy,xu2015secure,li2019privacyS,shoukry2016privacy,zhang2019admm}, which enables distributed computation while keeping individual inputs private. Despite its potential, SMPC often introduces significant communication overhead due to the need to split and distribute messages for secure transmission. Moreover, existing SMPC techniques primarily focus on consensus problems involving linear operations, such as averaging consensus, and their applicability to non-linear operations remains limited.  Another notable framework for addressing privacy concerns in distributed systems is subspace perturbation (SP) \cite{Jane2020ICASSP,Jane2020LS,li2020privacy}, along with its variants \cite{jordan2024,li2022communication,li2023adaptive}. SP is grounded in distributed optimization techniques, such as the Alternating Direction Method of Multipliers (ADMM) \cite{boyd2011distributed} and the Primal-Dual Method of Multipliers (PDMM) \cite{zhang2017distributed,heusdens2024distributed}. The core idea involves injecting noise into the non-convergent subspace of the optimization process. This approach ensures privacy preservation by obfuscating sensitive information while simultaneously preserving the accuracy and integrity of sensitive information in the convergent subspace.  Recent work \cite{yu2024privacy} has extended SP to maximum consensus, but its applicability to median consensus remains unexplored.

Although median consensus algorithms have demonstrated exceptional robustness against outliers, their privacy implications remain critically underexplored. This leaves a significant gap in understanding how to protect sensitive information while maintaining robustness and efficiency. In this work, we take the first step toward bridging this gap by investigating whether perfect privacy (zero information leakage) can be achieved while maintaining consensus correctness.  In particular, we establish a sufficient and necessary condition for achieving perfect privacy preservation in an L1-norm formulation of the median consensus problem. Rooted in an information-theoretic framework, our methodology derives conditions that guarantee zero leakage of private inputs for honest nodes. 
Experimental results confirm the practical feasibility of our derived conditions and illustrate how network density and convergence parameters impact privacy. To the best of our knowledge, this work presents the first rigorously analyzed privacy-preserving median consensus algorithm, supported by a comprehensive information-theoretic foundation. Our study bridges a critical gap between robustness and privacy in distributed consensus systems.

\section{Preliminaries}

\subsection{Problem formulation}\label{ss.pf}
We model the network as an undirected graph $G = (\mathcal{V},\mathcal{E})$, where $\mathcal{V} = \{1,\ldots,n\}$ represents the set of nodes/participants in the network and $\mathcal{E} \subseteq \mathcal{V} \times \mathcal{V}$ represents the set of edges, indicating the communication links between nodes. For each node $i$, its neighbor is defined as $\mathcal{N}_i = \{j \in{\cal V} \,|\, (i,j) \in \mathcal{E}\}$ and its degree is given by $d_i = |\mathcal{N}_i|$. Each node $i\in{\cal V}$ holds a local data value\footnote{For simplicity, $s_i$ is assumed to be a scalar, though the results can be easily generalized to vectors.} $s_i\in \mathbb{R}$. When we consider $s_i$ as a realization of a random variable, the corresponding random variable will be denoted by $S_i$ (corresponding capital). 

The privacy-preserving median consensus problem involves determining the median value $s_{\rm med}$ without revealing the private data $s_i$. This problem can be formulated \cite{deplano2023unified} as an L1-norm minimization problem: 
\begin{equation}
\begin{array}{ll} \text{minimize} & {\displaystyle \sum_{i\in {\cal V}} |x_i-s_i|,} \\\rule[4mm]{0mm}{0mm}
\text{subject to} & x_i -x_j = 0, \quad (i,j)\in \cal  E.
\end{array}
\label{eq:med_eq}
\end{equation}

\subsection{A/PDMM framework}

Building on the work of \cite{zhang2017distributed}, we can address the problem of minimizing a separable convex function under a set of equality constraints. The optimization problem is formulated as follows 
\begin{equation}
\begin{array}{ll} \text{minimize} & \displaystyle\sum_{i \in \mathcal{V}} f_i\left(x_i\right), \\\rule[3mm]{0mm}{0mm}
\text{subject to} & A_{ij} x_i+A_{ji} x_j =b_{ij}, \quad(i, j) \in \mathcal{E},
\end{array}
\label{eq:problem}
\end{equation}
where $f_i$ are convex, closed and proper (CCP) functions, $A_{ij}= -A_{ji} = 1$ for $i<j$, and $b_{ij}=0$. To solve \eqref{eq:problem}, the update equations for node $i\in \mathcal{V}$ are given by
\begin{align}
x_i^{(t)} &=\arg \min _{x_i}\left(\!\!f_i(x) +\sum_{j \in \mathcal{N}_i}\big( z_{i \mid j}^{(t)} A_{i j} x_i+\frac{c}{2}\|x_i\|^2\big)\!\!\right)\!\!,\\
    \nonumber y_{i \mid j}^{(t)}&=z_{i \mid j}^{(t)}+2 cA_{i j} x_i^{(t)}, \\
    \nonumber z_{i\mid j}^{(t+1)}&=
        (1-\theta)z_{i\mid j}^{(t)}+\theta y_{j\mid i}^{(t)},\label{eq.z_update} \rule[5mm]{0mm}{0mm}
\end{align}
where $y$ and $z$ are auxiliary variables, $\theta\in (0,1]$ is an avaraging constant and $c>0$ is a convergence 
parameter that influences the convergence rate. When $\theta=1$, the algorithm corresponds to  standard PDMM, while $\theta=\frac{1}{2}$ represents the $\frac{1}{2}$-averaged version of PDMM, which is equivalent to ADMM. In this work, since the L1 norm lacks uniform convexity,  standard PDMM will fail to converge. For this reason, we will focus our analysis on the case where $\theta=\frac{1}{2}$.

\subsection{Adversary model and evaluation metrics}\label{ss.eval}

\noindent\textbf{Adversary model}: 
We examine two commonly studied adversary models in the context of privacy-preserving algorithms. The first is the passive adversary model, which involves corrupt nodes within the network. These nodes adhere to the protocol, but they collide to collect and exchange information in an attempt to infer private data. We denote the set of corrupt nodes as $\mathcal{V}_c$ and the set of honest nodes as $\mathcal{V}_h$. Hence, $\mathcal{V}_h \cup \mathcal{V}_c = {\mathcal V}$ and $\mathcal{V}_h \cap \mathcal{V}_c = \emptyset$. The second model is the eavesdropping adversary, which intercepts all messages transmitted over unencrypted communication channels. These adversaries are assumed to collaborate, combining their knowledge to infer the private data of honest nodes. 



\noindent\textbf{Individual privacy}:
 In this work, we employ mutual information as a privacy metric due to its demonstrated effectiveness in prior research \cite{duchi2013local,kairouz2014extremal,li2022privacy}. Let $S_i$ be the random variable that represents the private data at node $i$, and $\mathcal{O}$ represent the total information observable by the adversary. The mutual information $I(S_i;\mathcal{O})$ \cite{cover2012elements} quantifies the amount of information about $S_i$ that can be inferred from $\mathcal{O}$. It is defined as
\begin{equation}
    \nonumber I(S_i;\mathcal{O})=h(S_i)-h(S_i\mid \mathcal{O}),
\end{equation}
where $h(\cdot)$ denotes differential entropy, assuming it exists.
When $I(S_i;\mathcal{O})=I(S_i;S_i)$, the adversary can fully reconstruct $s_i$. When $I(S_i;\mathcal{O})=0$, the adversary does not gain any information about $S_i$ by observing ${\cal O}$.

\section{Algorithm and Privacy Analysis}\label{sec.pa}
In this section, we first present the detailed A/PDMM algorithm designed to address the L1-norm formulation introduced in Section \ref{ss.pf}. Then, we analyze the privacy leakage in the corresponding distributed median consensus problem.

To address problem \eqref{eq:med_eq}, we can employ the regular equality constraint A/PDMM framework \cite{zhang2017distributed}. Specifically, we define $f_i(x_i)=|x_i-s_i|$.
This formulation aligns with the framework required for applying ADMM/PDMM. The detailed steps of this approach are outlined in Algorithm \ref{alg:med_primal}.

\begin{algorithm}[ht]
  \caption{Median consensus}
  \label{alg:med_primal}
  \begin{algorithmic}
     \ForAll{ $i \in \mathcal{V}, j \in \mathcal{N}_i$,}
     \State Randomly initialize $z_{i\mid j}^{(0)}$ \Comment{Initialization}
     \State $\text{Node}_j \leftarrow \text{Node}_i(z_{  i|j}^{(0)})$
     \EndFor
          \For{$t=0,1,...$} 
            \ForAll{$i \in \mathcal{V}$ } 
                    \State\vspace{-10pt}
                    \begin{align}\label{eq.x_update_eq}
     x_{  i}^{(t)}=\left\{ \begin{array}{ll}
         \frac{ -1- \sum_{  j \in {\cal N}_i} A_{  ij} z_{  i|j}^{(t)}}{cd_i},  &\text{if   } \frac{ -1- \sum_{  j \in {\cal N}_i} A_{  ij} z_{  i|j}^{(t)}}{cd_i}> s_i,\\
         \frac{ 1- \sum_{  j \in {\cal N}_i} A_{  ij} z_{  i|j}^{(t)}}{cd_i}, &\text{if   } \frac{ 1- \sum_{  j \in {\cal N}_i} A_{  ij} z_{  i|j}^{(t)}}{cd_i}< s_i,\\
         s_i, &\text{otherwise.} \rule[4mm]{0mm}{0mm}\\
    \end{array}\right.
\end{align}
                    \begin{equation}
                    \forall j\in \mathcal{N}_i: z_{  j|i}^{(t+1)}
        =\frac{1}{2}z_{j\mid i}^{(t)}+\frac{1}{2}(z_{  i|j}^{(t)}+2c A_{  ij}x_i^{(t)})\label{eq.z_update_eq}
                \end{equation}
                \State $\text{Node}_{j\in \mathcal{N}_i} \leftarrow \text{Node}_i(x_i^{(t)})$\Comment{Broadcast}
                \ForAll{ $j \in \mathcal{N}_i$} 
                 \State 
                 \begin{equation}\nonumber
                    z_{  j|i}^{(t+1)}
        =\frac{1}{2}z_{j\mid i}^{(t)}+\frac{1}{2}(z_{  i|j}^{(t)}+2c A_{  ij}x_i^{(t)})
                \end{equation}
            \EndFor
      \EndFor
      \EndFor
  \end{algorithmic}
\end{algorithm}


As defined in Section \ref{ss.eval}, we evaluate privacy preservation for node $i\in\mathcal{V}$ using the mutual information between $S_i$ and the total information observed by the adversary. The information observed throughout the iterations is given by $\{S_j,X_j^{(t)},Z_{j\mid k}^{(t)},Z_{j\mid k}^{(t)}\}_{j\in\mathcal{V}_c,k\in\mathcal{N}_j,t\in\mathcal{T}}$, where ${\mathcal T} = \{0,\ldots,t_{\rm max}\}$ and $t_{\rm max}$ is the maximum number of iterations. For the eavesdropping adversary, the information accessible in the communication channels includes $\{Z_{j\mid k}^{(0)},X_j^{(t)}\}_{j\in\mathcal{V},(j,k)\in\mathcal{E},t\in\mathcal{T}}$. By eliminating redundant information that both types of adversaries can access, individual privacy can be expressed as
\[
    I(S_i;\mathcal{O}) = I(S_i;\{S_j\}_{j\in\mathcal{V}_c},\{X_{j}^{(t)}\}_{j\in \mathcal{V},t\in\mathcal{T}},\{Z_{j\mid k}^{(t)}\}_{(j,k)\in\mathcal{E},t\in\mathcal{T}})
\]
We have the following result.
\begin{theorem}\label{thm.med_l1}
    Let $i\in{\cal V}_h$. We have
    \[
    I(S_i;\mathcal{O})=0,
    \]
    if and only if for all $t\in\mathcal{T}$,
    \begin{equation}
    s_i\notin[\frac{ -1- \sum_{  j \in {\cal N}_i} A_{  ij} z_{  i|j}^{(t)}}{cd_i},\frac{ 1- \sum_{  j \in {\cal N}_i} A_{  ij} z_{  i|j}^{(t)}}{cd_i}].
    \label{eq:cond}
  \end{equation}
  \label{theo:I}
    \end{theorem}
\begin{proof}
Replacing $z_{j\mid i}^{(t)}$ and $z_{i\mid j}^{(t)}$  in \eqref{eq.z_update_eq} we obtain 
\begin{align}  
z_{j\mid i}^{(t+1)}-z_{j\mid i}^{(t)}&=cA_{ij}x_i^{(t)}-\frac{1}{2}cA_{ij}x_i^{(t-1)}+\frac{1}{2}cA_{ji}x_j^{(t-1)}
\nonumber \\
&= cA_{ij} \left(x_i^{(t)} - \frac{1}{2}x_i^{(t-1)} - \frac{1}{2}x_j^{(t-1)}\right).
\label{eq:zdif} 
\end{align}
Assume condition \eqref{eq:cond} holds. Considering the difference $x_j^{(t+1)}-x_j^{(t)}$ using \eqref{eq.x_update_eq} and combining with \eqref{eq:zdif} we obtain
\begin{align} \label{eq:zxSuf}
    \nonumber x_j^{(t+1)}-x_j^{(t)} &=\frac{-\sum_{  k \in {\cal N}_j} A_{  jk} (z_{  j|k}^{(t+1)}-  z_{  j|k}^{(t)})}{cd_j}+r_j^{(t)}\\
    &=\frac{\sum_{  k \in {\cal N}_j}  (x_k^{(t)}-\frac{1}{2}x_k^{(t-1)}-\frac{1}{2}x_j^{(t-1)})}{d_j}+r_j^{(t)},
\end{align}
where $r_j^{(t)} \in \{0, \pm \frac{2}{c d_j}\}$ is a constant depending on the range of $x_j^{(t+1)}$ and $x_j^{(t)}$.
Hence,
    \begin{align*}
    I(S_i;\mathcal{O}) &=
    I(S_i;\{S_j\}_{j\in\mathcal{V}_c},\{X_{j}^{(t)}\}_{j\in \mathcal{V},t\in\mathcal{T}},\{Z_{j\mid k}^{(t)}\}_{(j,k)\in\mathcal{E},t\in\mathcal{T}})\\
    &\overset{(\rm a)}{=} I(S_i;\{S_j\}_{j\in\mathcal{V}_c},\{X_{j}^{(t)}\}_{j\in \mathcal{V},t\in\mathcal{T}},\{Z_{j\mid k}^{(0)}\}_{(j,k)\in\mathcal{E}})\\
    &\overset{(\rm b)}{=} I(S_i;\{S_j\}_{j\in\mathcal{V}_c},\{X_{j}^{(0)},X_{j}^{(1)}\}_{j\in \mathcal{V}},\{Z_{j\mid k}^{(0)}\}_{(j,k)\in\mathcal{E}})\\
    &\overset{(\rm c)}{=} I(S_i;\{S_j\}_{j\in\mathcal{V}_c},\{Z_{j\mid k}^{(0)}\}_{(j,k)\in\mathcal{E}})\\
    &\overset{(\rm d)}{=} 0,
\end{align*}
where (a) follows from \eqref{eq.z_update_eq}, (b) follows from \eqref{eq:zxSuf}, (c) holds since $\{x_{j}^{(0)},x_{j}^{(1)}\}_{j\in \mathcal{V}}$ can be computed from $z_{j|k}^{(0)}$ and $z_{k|j}^{(0)}$, while (d) holds since $S_j, j\neq i,$ and all $Z^{(0)}$s are independent of $S_i$.
\end{proof}


Two remarks are placed here: \newline
1) Theorem~\ref{theo:I} states that if node $i$ never transmits $s_i$ to its neighboring nodes, no information about $s_i$ will be disclosed. The fundamental principle of privacy preservation is thus to ensure that $s_i$ remains outside the interval
$
\left[ \frac{ -1- \sum_{j \in \mathcal{N}_i} A_{i j} z_{i \mid j}^{(t)}}{c d_i}, \frac{ 1- \sum_{j \in \mathcal{N}_i} A_{ij} z_{i \mid j}^{(t)}}{c d_i} \right]
$, which is of length $\frac{2}{c d_i}$. This length is solely determined by the degree of node $i$ and the convergence parameter $c$.  The impact of $d_i$ and $c$ on privacy preservation will be further discussed in Sections~\ref{ssec.di} and~\ref{ssec.c}, respectively. 

2)  The condition in equation~\eqref{eq:cond} is equivalent to requiring that $x_i^{(t)} \neq s_i$ for all $t \in \mathcal{T}$. Therefore, by monitoring whether the updated value $x_i^{(t)}$ ever equals the private value $s_i$, we can directly determine whether perfect privacy is maintained. 


\section{Simulation Results}
In this section we present numerical results to validate our theoretical claims and discuss the factors that influence privacy leakage.

\noindent\textbf{Initialization:} 
We generate a random geometric graph (RGG) \cite{penrose2003random} with $n=5$ nodes. The private data are randomly drawn from a standard normal distribution $\mathcal{N}(0,1)$. For the purpose of comparison, we fixed the values of $s_i$ throughout the experiments. However, it is important to emphasize that $s_i$ can, in fact, be randomized arbitrarily without affecting the validity of the analysis or the results. The auxiliary variables $z_{i\mid j}^{(0)}$ are sampled from a normal distribution $\mathcal{N}(\mu A_{ij},\sigma^2)$. The choice of $\mu$ plays a critical role in determining the initialization of $x$.

\subsection{Feasibility of derived conditions}
We investigate whether the derived condition in Theorem \ref{thm.med_l1} is satisfied in real-world scenarios.
We first set $\mu=0$ and $\sigma^2=10^{-2}$, ensuring that each node starts its iterations from a point near the median value. Figure \ref{fig:eq2} illustrates the convergence of the entire network (left plot) and individual nodes (right plot). Note that the actual private data values $s_i$ are indicated in the legend.
As shown in Figure \ref{fig:eq2}(b), the optimization variables of all nodes, except for the one holding the median value, satisfy $x_i^{(t)} \neq s_i$ for all $t \in \mathcal{T}$. Consequently, the condition in equation~\eqref{eq:cond} is automatically satisfied for these nodes, implying that perfect privacy is maintained throughout the optimization process.  


\begin{figure}[ht]
    \centering
\includegraphics[width=0.54\textwidth]{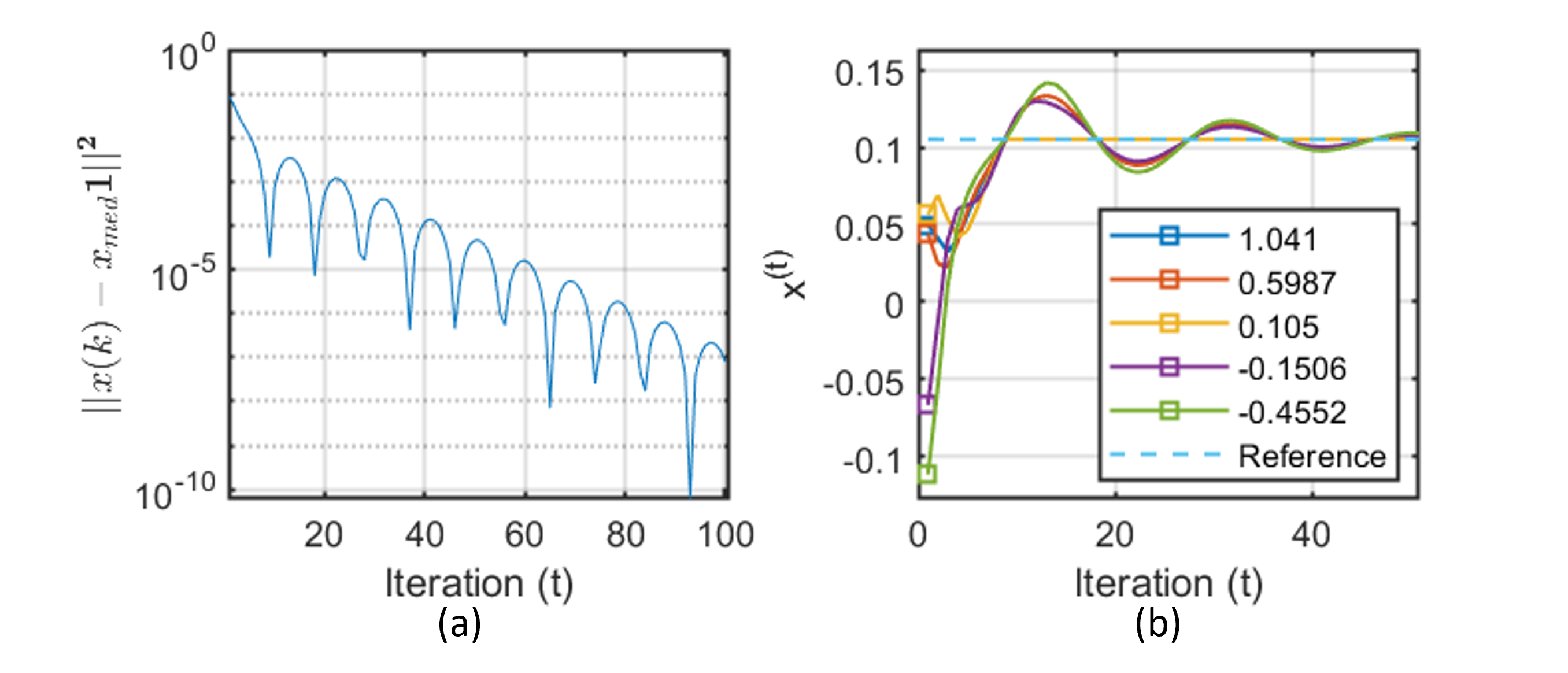}
    \caption{(a) Overall convergence error of the whole network; (b) Convergence of the optimization variable $x^{(t)}$ for each node,  as a function of iteration ($t$).}
    \label{fig:eq2}
\end{figure}


\subsection{The case when the conditions are not satisfied}
Although we demonstrated above that the derived conditions can be satisfied, there are scenarios in which they are not. For example, it is not always feasible to identify an initial point that lies sufficiently close to the median value without prior knowledge. 
Figure~\ref{fig:dp}(a) illustrates a scenario where $\mu = -10$ and $\sigma^2 = 1$. Under these conditions, the initialization of $x$ tends to be significantly larger than the actual data distribution. As shown, nodes holding values above the median (represented by the blue and red lines) may encounter a plateau during the process, where $s_i\in[\frac{ -1- \sum_{  j \in {\cal N}_i} A_{  i j} z_{  i|j}^{(t)}}{cd_i},\frac{ 1- \sum_{  j \in {\cal N}_i} A_{  i j} z_{  i|j}^{(t)}}{cd_i}]$. Consequently, privacy guarantees are only achievable for half of the nodes whose values lie below the median. In this case, only half of the nodes satisfy the privacy condition. To further enhance the privacy of nodes lacking guarantees, we can apply differential privacy techniques~ \cite{kefayati2007secure,wang2018differentially}. Adding differential private noise to our approach will result in the following information leakage:  
\begin{align*}
    I(S_i; \mathcal{O}) &= I(S_i; S_i + N_i) < I(S_i; S_i),
\end{align*}
where the random variable $N$ represents the introduced random noise. In Figure~\ref{fig:dp}(b), we illustrate the same scenario as in Figure~\ref{fig:eq2} and \ref{fig:dp}(b), but with the inclusion of a random offset drawn from a normal distribution \( \mathcal{N}(0, 10^{-2}) \). It is evident that the plateau still exists, however, noise offsets are introduced in the private value, thereby providing a certain level of privacy protection. This improvement in privacy comes at a cost: a reduction in accuracy, illustrating the inherent trade-off between privacy and precision.

\begin{figure}[ht]
    \centering
\includegraphics[width=0.4\textwidth]{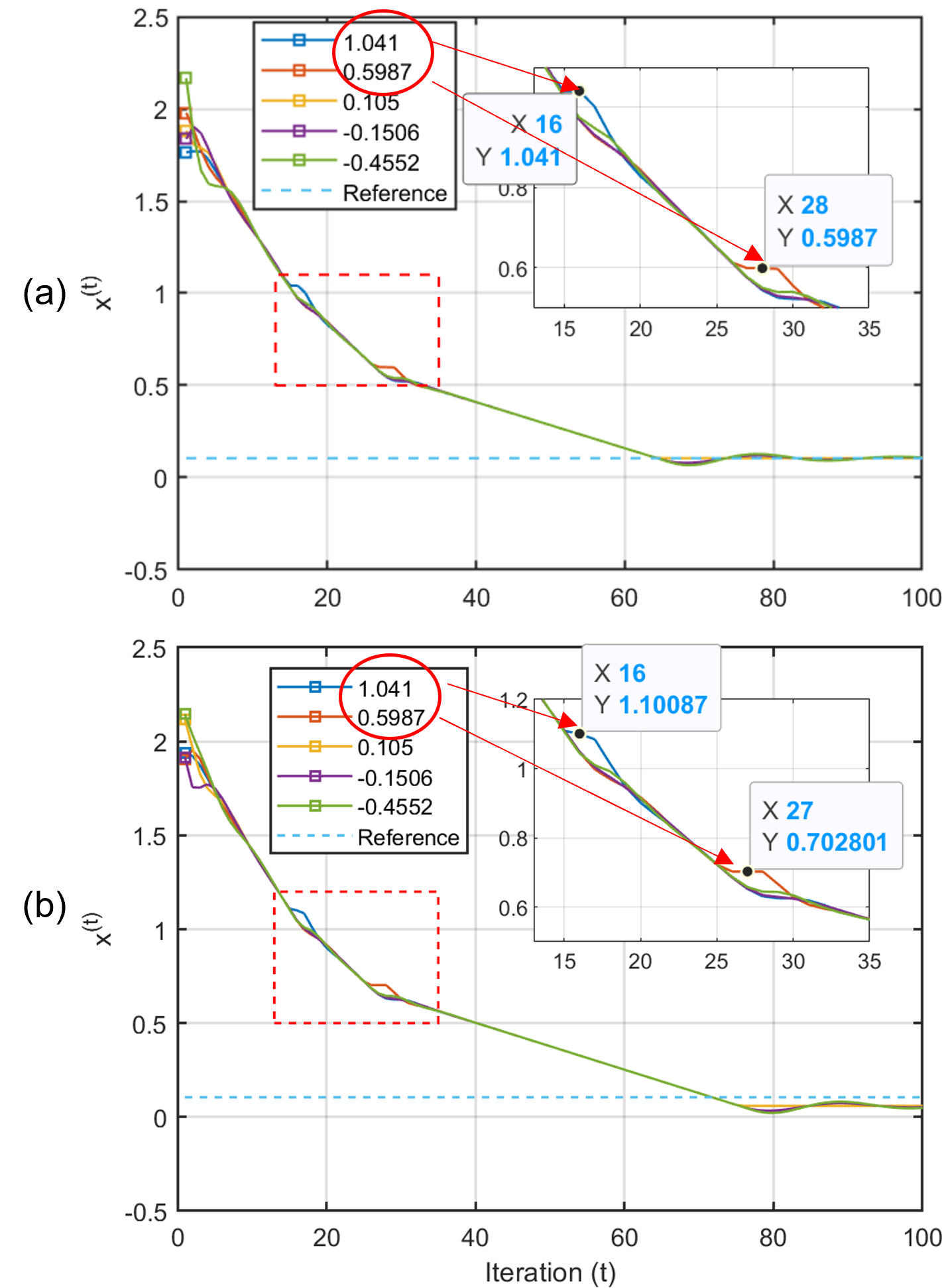}
    \caption{Information leakage with specific initialization: (a) without differential privacy and (b) with differential privacy.}
    \label{fig:dp}
\end{figure}

\subsection{The impact of topology density}\label{ssec.di}
As discussed in Section \ref{sec.pa}, the number of neighbors (degree $d_i$) influences the length of the decision interval, which affects individual privacy. To examine this effect, we generate three types of graphs with $n=15$ nodes: a ring graph (where each node has degree $d_i=2$), a complete graph (where each node has degree $d_i=|\mathcal{V}|-1$), and a random geometric graph (RGG). For each graph, we run 100 simulations and compute the average proportion of nodes that preserve their privacy. The results are presented in Figure~\ref{fig:d}. Our analysis shows that a denser topology, with a higher degree $d$, reduces the window width, allowing more nodes to securely preserve their values.
\begin{figure}[ht]
    \centering
\includegraphics[width=0.35\textwidth]{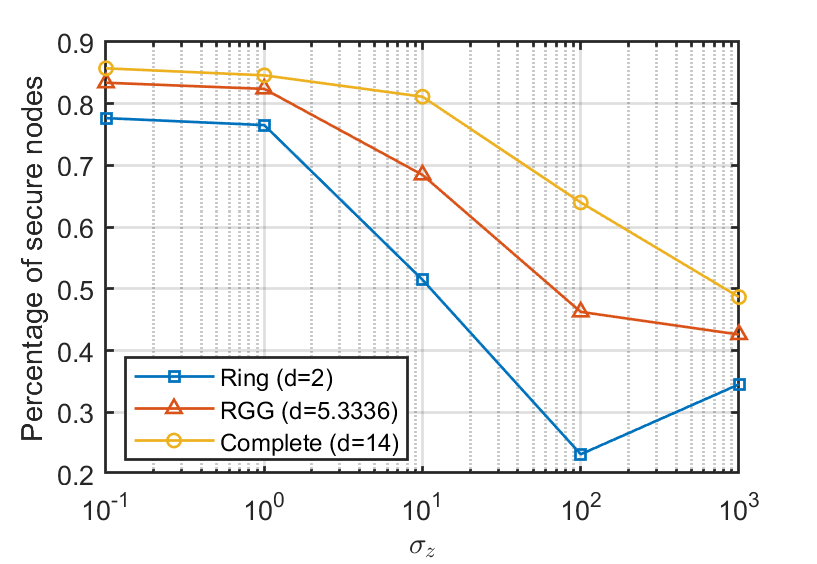}
    \caption{Proportion of secure nodes under different topologies (ring graph, random geometric graph and complete graph)}
    \label{fig:d}
\end{figure}

\subsection{The impact of convergence parameter $c$}\label{ssec.c}
In addition to the degree, the convergence parameter $c$ also plays a critical role in determining the convergence rate and the length of the decision interval. As illustrated in our analysis, the impact of this parameter on privacy mirrors the influence of node degree in Figure~\ref{fig:c}: as expected, a larger value of $c$ results in a smaller window, thus enhancing the ability of more nodes to securely preserve their values. This relationship highlights the need to carefully choose the convergence parameter to balance efficiency and privacy.
\begin{figure}[ht]
    \centering
\includegraphics[width=0.35\textwidth]{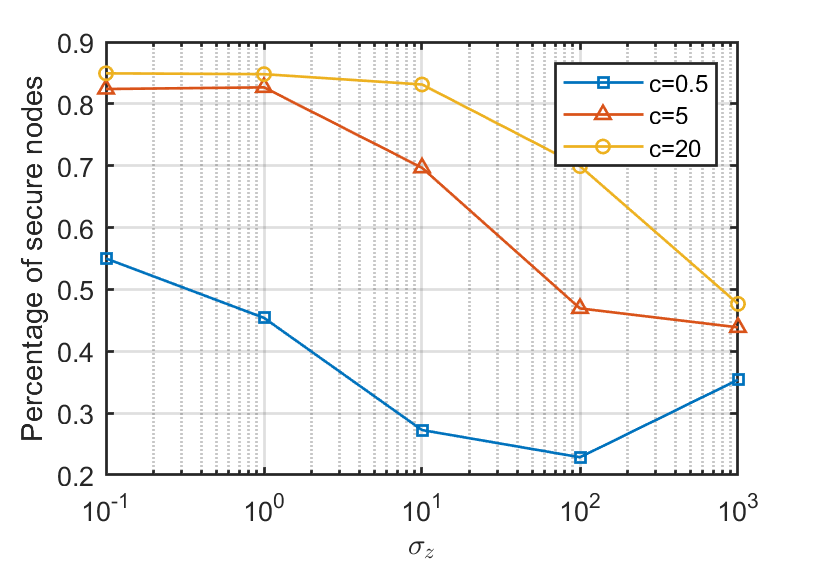}
    \caption{Proportion of secure nodes with different choices of $c$}
    \label{fig:c}
\end{figure}

\section{Conclusion}
In this study, we established sufficient and necessary conditions for achieving privacy-preserving median consensus under the ADMM/PDMM framework. Our results demonstrate that it is possible to achieve both perfect privacy and optimal accuracy simultaneously. In addition, we analyzed the influence of network density and convergence parameters on the robustness of privacy guarantees. Finally, we showed that incorporating differential privacy can provide additional protection when the derived conditions are not entirely met, albeit at the cost of reduced accuracy.



\bibliographystyle{IEEEbib}
\bibliography{refs}

\end{document}